\documentclass[12pt,letterpaper]{amsart}
\usepackage{amsmath}
\usepackage{setspace}
\doublespace
\usepackage[ansinew]{inputenc}

\usepackage{fixltx2e}
\usepackage{textcomp}
\usepackage{fullpage}
\usepackage{amsfonts}
\usepackage{verbatim}
\usepackage[english]{babel}
\usepackage{pifont}
\usepackage{color}
\usepackage{lscape}
\usepackage{indentfirst}
\usepackage[normalem]{ulem}
\usepackage{booktabs}
\usepackage{natbib}
\usepackage{float}
\usepackage{latexsym}
\usepackage{url}
\usepackage{hyperref}
\usepackage{epsfig}
\usepackage{graphicx}
\usepackage{amssymb}
\usepackage{bm}
\usepackage{array}
\usepackage{mhchem}
\usepackage{ifthen}
\usepackage{caption}
\usepackage{hyperref}
\usepackage{amsthm}
\usepackage{amstext}
\usepackage{subfig}
\usepackage{yfonts}
\usepackage{epic}
\usepackage{latexsym}
\usepackage{epsfig}
\usepackage[mathscr]{euscript}
\usepackage[all]{xy}

\raggedright
\renewcommand{\section}[1]{%
\bigskip
\begin{center}
\begin{Large}
\normalfont\scshape #1
\medskip
\end{Large}
\end{center}}

\renewcommand{\subsection}[1]{%
\bigskip
\begin{center}
\begin{large}
\normalfont\itshape #1
\end{large}
\end{center}}

\renewcommand{\subsubsection}[1]{%
\vspace{2ex}
\noindent
\textit{#1.}---}

\renewcommand{\tableofcontents}{}

\bibpunct{(}{)}{;}{a}{}{,}  

\newtheorem{thm}{Theorem}

\DeclareMathAlphabet{\mathpzc}{OT1}{pzc}{m}{it}

\newtheorem{lemo}[thm]{Lemma}
\newtheorem{proposition}[thm]{Proposition}  

\newcommand{\cp}{{\mathcal P}}

\newcommand{\RR}{{\mathbb R}} 

\newcommand{\cP}{{\mathcal P}}

\vspace{1.5in}

\begin{document}
\begin{flushright}
Version dated: \today
\end{flushright}
\bigskip
\noindent Axiomatic properties of species tree inference

\bigskip
\medskip
\begin{center}

\noindent{\Large \bf Axiomatic opportunities and obstacles for inferring a species tree  from gene trees}
\bigskip



\noindent {\normalsize \sc Mike Steel$^1$, and Joel D. Velasco$^{2}$}\\
\noindent {\small \it 
$^1$Allan Wilson Centre for Molecular Ecology and Evolution, University of Canterbury, Christchurch, New Zealand}\\
$^2$Texas Tech University, Department of Philosophy, Box 43092, Lubbock, TX 79409, USA\\
\end{center}
\medskip
\noindent{\bf Corresponding author:} Mike Steel, Mathematics and Statistics, University of Canterbury, Christchurch, New Zealand E-mail: mike.steel@canterbury.ac.nz\\


\vspace{1in}

\newpage
\subsubsection{Abstract}The reconstruction of a central tendency `species tree' from a large number of conflicting gene trees is a central problem in systematic biology. Moreover, it becomes particularly problematic when taxon
coverage is patchy, so that not all taxa are present in every gene tree. Here, we list four apparently  desirable properties that a method for estimating a species tree from gene trees could have (the strongest property states that building a species tree from input gene trees and then pruning leaves gives a tree that is the same as, or more resolved than,
the tree obtained by first removing the taxa from the input trees and then building the species tree).   We show that while it is technically possible to simultaneously satisfy
these properties when taxon coverage is complete, they cannot all be satisfied in the more general supertree setting.  In part two, we discuss a concordance-based consensus method based on Baum's `plurality clusters', and an extension to concordance supertrees. \\


\vspace{1.5in}

\section{Introduction}
  
Reconstructing a rooted phylogenetic species tree from a collection of gene trees (one for each genetic locus) can be viewed as a type of voting procedure.  Each locus supports a gene tree  and  tree reconstruction seeks to return a species tree based on the relative support of different trees from the population of voters (trees).  In social choice theory, 
 Arrow's theorem \citep{arr} has long played a prominent role. This theorem demonstrates that seemingly reasonable and desirable criteria for converting individual rankings of candidates into a community-wide ranking of candidates cannot be simultaneously satisfied.    In  phylogenetics, similar questions arise as to whether methods exist for combining trees so as to satisfy desirable properties (`axioms'), and a number of authors have shown that various combinations of axioms are impossible \citep{bar, bar2, day, mcm, mcm2, steel, tha}.  

In this short note, we describe some further results based on a slightly different set of assumptions that are appropriate to settings where taxon coverage across loci can be patchy \citep{san}, and where consensus methods (which require complete taxon coverage) must be replaced by more general supertree approaches.  We show that certain axioms can be satisfied in the consensus setting where taxon coverage is complete across loci, though by somewhat contrived consensus methods rather than the standard ones in common use.  However when we move to the supertree setting, where taxon coverage can be incomplete across loci,  an Arrow-type obstacle arises. We briefly discuss the biological implications of these two results, and  then consider the properties of a particular consensus tree approach (`plurality clusters' of \cite{bau, bau2})  and how this could  be extended to the supertree setting.

\subsection{Axioms for reconstructing a species tree from gene trees}

Formally, a {\em species tree estimator} is a  function  $\psi$ that assigns a rooted phylogenetic $X$-tree to any profile (i.e. sequence) $\cp=(T_1,\ldots, T_k)$ of trees at different loci, where $X$ is the set of taxa that occur in at least one tree.  Throughout this paper, all trees are rooted phylogenetic trees, and so each tree can be thought of as a hierarchy (i.e. a collection of subsets of the nonempty leaf set $Y$, containing $Y$ and the singletons $\{y\}: y \in Y$, and satisfying the nesting property that any two sets are either disjoint or one is a subset of the other). 

If each tree $T_i$ has the same leaf set $X$ then $\psi$ constructs a {\em consensus tree},
while if the leaf set of the trees $T_1,\ldots, T_k$ are not all equal to $X$ (due to patchy taxon coverage across loci) then $\psi$ constructs a {\em supertree}.
A  tree reconstruction procedure is regarded as fully deterministic (e.g. in the case of ties, as with equally most parsimonious trees, one might take the strict consensus of the resulting trees). 

In order to state the four axioms we first make two key definitions. Given a rooted phylogenetic $X$-tree $T$ and a subset $W$ of taxa let $T|W$ denote the rooted phylogenetic tree
that $T$ induces on the leaf set $X \cap W$. That is, $T|W$ is the rooted phylogenetic $(X \cap W)$-tree obtained from $T$ by taking the minimal subtree of $T$ that connects the leaves
in $X \cap W$ and then suppressing any vertices that have just one outgoing arc.  Notice that $W$ need not be a subset of $X$, and $W$ could even be disjoint from $X$ in which case $T|W$ is the empty set.  Given a profile (sequence) of trees $\cp = (T_1, \ldots, T_k)$ and any subset $W$ of $X=\cup_{i=1}^k X_i$, let $\cp|W =(T_1|W, \ldots, T_k|W)$.  In case one of these trees is the empty set, we will delete it from the profile, while retaining the ordering of the remaining trees, to obtain a shorter profile. 

Consider then a profile $\cp = (T_1, \ldots, T_k)$ (for any $k \geq 1$), of rooted phylogenetic trees, where $X_i$ is the leaf set of $T_i$ for each $i$, and $X=\cup_{i=1}^k X_i$, and the following four
conditions: 
\begin{itemize}
\item[{\bf (A1)}] (`Unrestricted domain') In the consensus setting (i.e. $X_i=X$ for all $i$), or the more general super-tree setting (where the $X_i$ are allowed to differ) $\psi(\cp)$  is a  rooted phylogenetic $X$-tree for any choice of $\cp$. 
\item[{\bf (A2)}]   (`Unanimity') When $\cp = (T, T, \ldots, T)$ for a fixed tree $T$, then  $\psi(\cp) = T$.
\item[{\bf (A3)}] (An `Irrelevance' axiom) 
Suppose that  $X_i$ contains just one or two taxa, each of which is present in at least one of the other $k-1$ taxon sets.   Let $\cp'$ be the profile obtained by removing $T_i$ from $\cp$.  Then  $\psi(\cp')= \psi(\cp)$. 
\item [{\bf (A4)}] (A weak `Independence' condition).   
For any non-empty subset $Y$ of $X$ the tree $\psi(\cp|Y)$ coincides with, or is refined by, the tree $\psi(\cP)|Y$.
\end{itemize}

In words, (A3) says that if one rooted tree has at most two taxa, and these are already present in other taxon sets, then removing this tree should not alter the relationships of taxa on the final tree returned by the method.   The idea here is that such a trivial tree carries no relevant phylogenetic information, so it should not affect the outcome of the method.

Condition (A4) states that if we build a species tree from input gene trees and then prune some of the leaves, the resultant tree should be the same as, or perhaps a more resolved version of, the species tree we would obtain by first removing those taxa from the input gene trees and then building our species tree.

\section{Results for consensus trees and supertrees}

In the consensus setting, where all the input trees have the same leaf set, condition (A3) holds vacuously, provided that at least three leaves are present. In this case,  examples of methods that satisfies (A1) and (A2) include  the `strict consensus' and `majority rule consensus' tree (i.e. the tree that contains the clusters present in all trees in the profile, or in a majority of the trees in the profile, respectively).  However, each of these methods fail condition (A4).  
Fig.~1 shows why: If we take as input the two trees $((xa)b)(cy)$ and $((ya)b)(cx)$, their consensus (under strict or majority rule consensus) will be the star tree $(a,b,c,x,y)$.
However,  if we were to restrict each input tree  to the leaf subset $\{a,b,c\}$, the two  induced input trees would both become $(ab)c$, and so too would their consensus.  But $(ab)c$ is neither equal to, nor revolved by the (star) tree obtained from the star tree $(a,b,c,x,y)$  by restricting to leaf set $\{a,b,c\}$.  This example also shows that some other consensus methods, including  the $R^*$ method described by \cite{bry}, do not satisfy (A4).

 \begin{figure}[h]
  \begin{center}
\resizebox{12cm}{!}{
\includegraphics{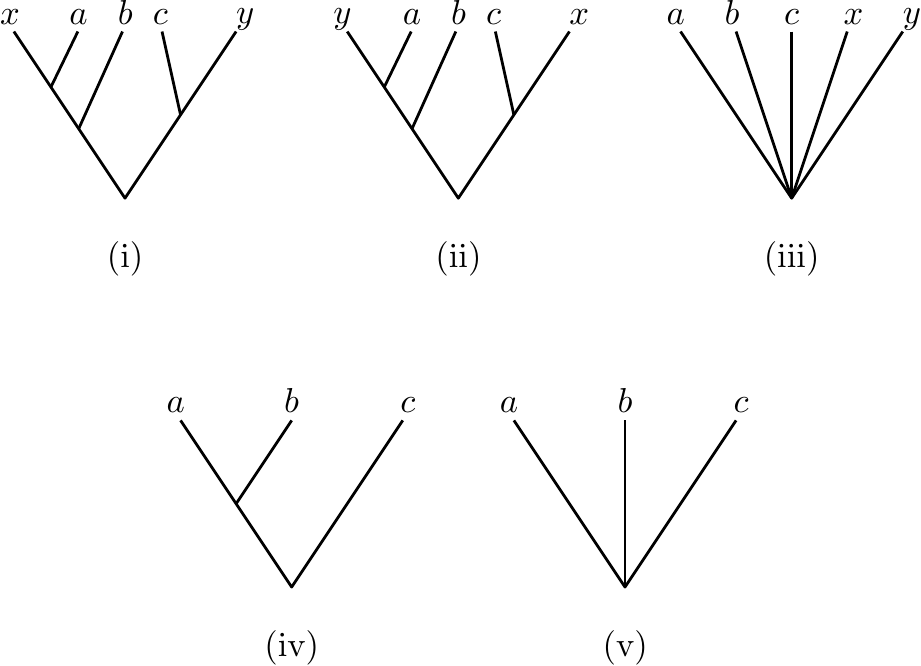}
}
\caption{A case where (A4) fails for consensus methods such as strict or majority rule consensus.   For the trees  $T_1$ and $T_2$ shown in (i) and (ii), their  consensus tree is the star tree shown in (iii). But restricting each input tree to $\{a,b,c\}$ results in the consensus tree shown in (iv), which is
neither the same as, nor resolved by the tree (v) obtained by restricting the output tree in (iii) to the taxon subset $\{a,b,c\}$.}
\end{center}
\end{figure}

Nevertheless, in this consensus setting, there exist methods that satisfies all four properties (A1)--(A4), a simple example of one is the method $\psi_1$ that simply returns the first tree in the profile. There are other methods also (not just projections onto the tree appearing at a given position in the profile). For example, consider the following simple modification of $\psi_1$, denoted $\psi_1^*$:   Given a profile $\cp=(T_1, \ldots, T_k)$, if there are three elements $a,b,c \in X$ for which $ab|c \in T_i|\{a,b,c\}$ for every $i=1, \ldots, k$, then let 
$\psi_1^*(\cp)=\psi_1(\cp)=T_1$; otherwise,  let $\psi_1^*(\cp)$ be the completely unresolved star tree. 

While $\psi_1$ and $\psi_1^*$ suffice to show that (A1)--(A4) can be satisfied in the consensus setting, these particular two methods have an undesirable property that would make them
quite unsuitable in practice, namely the output tree can depend on the order of the input trees.  It thus seems an interesting question as to whether there is a consensus method $\psi$ which, in addition to (A1)--(A4), always outputs the same tree regardless of the ordering of the trees in the profile $\cp$  (i.e. $\psi(T_1,\ldots, T_k) = \psi(T_{\sigma(1)}, \ldots, T_{\sigma(k)})$ for any permutation $\sigma$ of $\{1, \ldots, n\}$). We had initially believed that the Adams consensus method \citep{ada, bry} which satisfies this last property, along with (A1)--(A3) would also satisfy (A4); however it does not,  as the following simple example (due to R. C. Powers) shows. Consider two rooted trees $T_1$ and $T_2$ on leaf set $X=\{a,b,c,d,e,f\}$, where $T_1$ has as its nontrivial clusters $\{a,b,c,e\}$ and $\{a,b,c,d,e\}$, while $T_2$ has as its nontrivial clusters $\{b,c,d\}, \{b,c,d,e\}$ and $\{a,b,c,d,e\}$.
Let $Y = \{b,c,d,e\}$.  Then the Adams consensus tree of $(T_1|Y, T_2|Y)$ has the non-trivial cluster $\{b,c\}$, while if we take the Adams consensus of $T_1$ and $T_2$ and restrict to $Y$ then the resulting tree does not contain $\{b,c\}$ as a cluster.

\subsection{An Arrow-type impossibility result in the supertree setting}

In the supertree setting, it is also easy to find methods that simultaneously satisfy (A1), (A3) and (A4); a trivial example is the method that constructs
the completely unresolved star tree for all inputs.  

Satisfying (A1), (A2) and (A3) together is also fairly straightforward -- output the star tree unless, for some tree $T$, the input trees $(T_1,\ldots, T_k)$ have the property 
that $T_i = T$ for all $i$ in some nonempty subset $I$ of $\{1, \ldots, k\}$, and $T_j$ is a tree with just one or two leaves for all $j \in \{1, \ldots, k\} - I$; in which case we output the tree obtained from $T$ by attaching any  leaf (or leaves) in $X- \cup_{i \in I}X_i$ so that they are adjacent to the root.

Is there a method that satisfies all four properties (A1)--(A4) in the super tree setting? No. Even if we weaken (A1) to:

\begin{itemize}
\item[{\bf (A1$^-$)}] For any profile $\cP$ of rooted phylogenetic trees, $\psi(\cP)$is a rooted phylogenetic tree  on {\em all or some  of} the  taxa mentioned by the input trees;
\end{itemize}
(which allows us to delete taxa from the supertree if necessary) our main result shows that no such method can simultaneously accommodate these conditions.
Formally, we have the following result, whose proof is provided in the Appendix.

\begin{proposition}
\label{main}
No tree reconstruction procedure exists that simultaneously satisfies axioms (A1$^-$), (A2), (A3) and (A4) on all inputs. 
\end{proposition}

\subsection{Biological significance}

Suppose we have a fixed set $S$ of species. 
It is clear that (even in the consensus setting) any method for building a species tree from gene trees should allow the tree to change as more loci are sequenced and the gene trees for these loci are  included in the analysis (since the gene trees at later loci may, for example favour a different species tree). 

But suppose we fix the set of available loci, and instead try to build a tree by adding taxa.  We might try and construct a tree for some of the taxa and then sequentially try to attach each additional taxon in an optimal place in this tree. On occasions, an additional taxon may even allow us to resolve
the tree a bit better, but we do not wish to go back and rearrange the consensus or supertree we obtained at an earlier stage of the process. If our method satisfied (A4), we would be able to do this step-by-step construction.   When we can determine a tree for each locus on all the taxa we have at any given stage, the goal is achievable (since in that setting there are methods that satisfy (A1)--(A4)), albeit by somewhat contrived consensus methods, rather than commonly-used ones.   But where taxon coverage is inherently patchy between loci, Proposition~\ref{main} dashes any hope of a general method that would guarantee to  achieve that goal while also satisfying the clearly desirable properties (A1), (A2), and (A3).  

Violating (A4) means that adding, for example, taxon $d$ to the set can change how $a$, $b$, and $c$ are related to each other. So for example, $d$ might be closest to $a$ on some gene trees and closest to $c$ on others indicating that perhaps $a$ and $c$ are closer than previously thought. 

There is an important distinction to be made here. Imagine that building a species tree consists of moving from data to gene trees and then from gene trees to a species tree. One way in which adding a taxon can change the relationships between other taxa is by changing the gene trees themselves. So for example, at locus $i$ we might have the tree $((ab)c)$ while after adding taxon $d$ we might get $(((ad)c)b)$. Evidence for relationships is `holistic' in this way. Thus building a tree and then pruning some leaves does not necessarily yield the same tree as removing those taxa from the sequence data and then building the tree. This is a fact about the relationship between input data and gene trees. But this difference is perfectly consistent with satisfying (A4). (A4) asserts something about the relationship between gene trees and a species tree constructed from them. It says that pruning taxa from the gene trees can't change the relationships of the other taxa on the resulting species trees.

We want it to be the case that adding or removing taxa from the sequence data can change the relationships between other taxa on gene trees. It is less clear, but upon reflection also true that we do want it to be possible that adding or removing taxa from gene trees can change the relationships between other taxa on the resulting species trees. Thus axiom (A4) is too strong. One consequence of this is that the best methods that take gene tree inputs and output a species tree cannot proceed by simply adding one taxon at a time to an already existing tree but rather must consider all of the taxa at once.


\section{Concordance trees and plurality consensus}
In the second part of this note, we move from general axiomatic considerations
to the study of a particular class of consensus and supertree methods that are based
on the frequency of clades amongst the input trees. The methods considered in this
section will satisfy (A1), (A2) and (A3), but not (A4).

For the consensus setting, given a sequence $(T_1, T_2, \ldots, T_k)$ of rooted phylogenetic $X-$trees, define the {\em concordance factor} of any subset $A$ of $X$, denoted $cf(A)$, to be
the proportion of trees that contain $A$ as a cluster. We say that two subsets $A$ and $B$ of $X$ {\em overlap} if $A \cap B$ is a nonempty strict subset of $A$ and of $B$ (this is equivalent to $A$ and $B$ being incompatible in the sense that no tree could contain both sets as clusters). 
Concordance factors form the basis of some well-studied consensus methods, including:

\begin{itemize}
\item 
{\bf Strict consensus:} The tree having as its clusters those subsets $A$ of $X$ with $cf(A)=1$.

\item {\bf Majority consensus:} The tree having as its clusters those subsets $A$ of $X$ with $cf(A)>0.5$.

\item{\bf Majority (+) consensus:}   The tree having as its clusters those subsets  $A$ of $X$ with  $cf(A)>cf^{o}(A)$, where $cf^{o}(A)$ denotes the proportion of input trees that have a cluster that $A$ overlaps.

\item{\bf Greedy consensus:}  The  tree obtained by ranking the clusters present in the input trees according to their concordance factor, and constructing a set of clusters, beginning with the cluster of maximal $cf$-value, and adding further clusters in the order of their diminishing $cf$-values, omitting any clusters that overlap with any of the clusters so far accepted. Ties are broken arbitrarily. 

\end{itemize}

The literature on consensus methods is vast, with \citet{bry} providing a helpful survey. Strict and majority consensus trees are well studied, while the majority (+) consensus tree is more recent, with a mathematical analysis by \citet{don2} revealing how this approach can be characterized as a type of consensus median method.   The majority (+) consensus clusters always form a hierarchy (i.e. a tree) and this hierarchy contains the majority clusters, which in turn contain the strict clusters.

Whether a set of taxa is a majority (+) cluster depends on the proportion of input trees containing clades that contradict the set in question.  If all the  input trees  are fully resolved (binary) then $cf^{o}(A) = 1 - cf(A)$ and thus the majority (+) clusters are exactly the majority clusters.  However, when one or more of the trees is not fully resolved, there may exist majority (+) clusters that are not majority clusters (a simple example is provided by the two trees $((ab)c)$ and $(abc)$). 

Greedy consensus differs from the other methods in that the resulting consensus tree is not uniquely specified, since the possibility of ties means that one can obtain different trees according to how such ties are broken.  This would generally be regarded as an undesirable property,  since we would like a consensus method to output a tree that is independent of arbitrary choices.

 \cite{bau, bau2}  introduced the quantitative notion of a `concordance factor' in phylogenomics  -- `the proportion of the genome for which a given clade is true', and  `plurality' to  mean  that the concordance factor of the clade is higher than the concordance factor of any contradictory grouping.     In other words, Baum considers subsets $A$ of $X$ for which $cf(A) > cf(B)$ for every subset $B$ of $X$ which overlaps with $A$, and this notion has also been referred to as `frequency-difference' clusters by others (\cite{gol, don2, jan}), so we will mostly use this terminology here.   A crucial point is that frequency-difference clusters form a hierarchy and thereby a tree, namely the  {\em frequency-difference consensus}  tree, proposed by Pablo Goloboff, and implemented in his phylogenetic software package TNT \citep{gol}.   More recently, efficient algorithms for constructing the frequency-difference consensus tree and the majority (+) tree have been described in \cite{jan}.  

Fig.~2 shows an example of four distinct input trees (i)--(iv), for which the frequency-difference consensus tree (v) has a nontrivial cluster, while the majority (+)  (and so majority 
and strict) consensus trees is the unresolved star tree (vi).

 \begin{figure}[h]
  \begin{center}
\resizebox{12cm}{!}{
\includegraphics{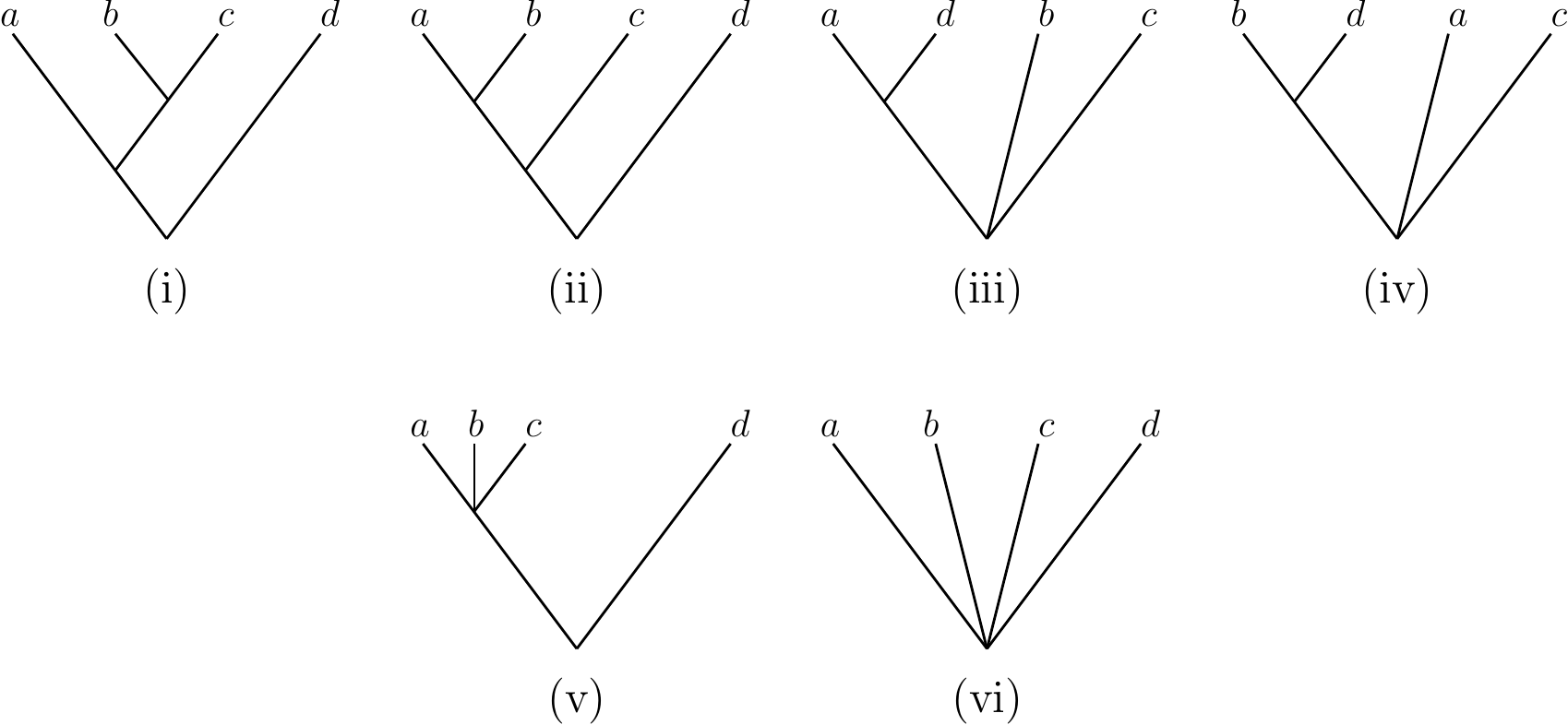}}
\caption{
(i)--(iv) Four input trees on the same leaf set; (v) the frequency-difference consensus tree, (vi) the majority (+) consensus tree.
}
\end{center}
\end{figure}

The relationship between frequency-difference consensus and the other consensus methods was investigated by \cite{don2}, and can be summarised as follows:
\begin{quote}
The frequency-difference consensus tree refines the majority (+) consensus tree (and thereby also the strict and majority consensus tree), and this refinement can be proper. 
In turn, the frequency-difference consensus tree is refined by any greedy consensus tree and this refinement can also be proper. 
\end{quote}
To see this, observe that the majority consensus tree is a (possibly proper) refinement of strict consensus and we described above how the majority (+) tree is a (possibly proper) refinement of the majority consensus tree. Now suppose that $A$ is a majority (+) cluster, so that $cf(A) > cf^{o}(A)$.  Let $B$ be any cluster that overlaps $A$.  Then $cf^{o}(A) \geq cf(B)$, by definition.
Thus, $cf(A) > cf^{o}(A) \geq cf(B),$ and so $cf(A)>cf(B)$. Since this holds for each cluster $B$ that overlaps $A$ it follows that $A$ is a frequency-difference cluster. Thus any majority (+) cluster is also a frequency-difference cluster, and so the  frequency-difference consensus tree refines the majority (+) consensus (and thereby also the majority and strict consensus) tree.

The example in Fig.~2 shows that the frequency-difference consensus tree can be strictly more resolved than the  majority (+) consensus tree (an example involving  three taxa is also possible if we allow an input tree to occur more than once: for input trees $(ab)c, (ab)c, (ac)b$ and $(bc)a$, the set $\{a,b\}$ is a frequency-difference cluster, but not a majority (+) cluster). 

Regarding the relationship between frequency-difference and greedy consensus methods,  suppose that $A$ is  a frequency-difference cluster.  Then $A$ appears higher in the ranking of $cf$-values than any overlapping cluster, and so it must be contained in every greedy consensus tree. The example of input trees of $(ab)c$ and $a(bc)$ suffices to show that a greedy consensus tree can be a proper refinement of the frequency-difference consensus tree.  Thus these consensus trees are introduced above in an order such that each tree is a refinement (and possibly a proper one) of each of the trees above it, with frequency-difference consensus fitting between the majority (+) and greedy consensus methods.

Notice that the consensus methods satisfy properties (A1), (A2) and (trivially)
(A3), but can fail (A4).  For example, the same profile used in Fig.~1 to show that strict consensus fails to satisfy (A4) also applies to the frequency-difference consensus tree (which is also a star tree for this profile). 

One advantage of frequency-difference consensus over (say) majority  consensus or strict consensus is that it avoids setting a particular threshold for concordance factors to reach (such as 0.5 in the case of majority consensus), which is mathematically convenient but would seem to have little biological rationale.  
We plan to discuss  the biological and philosophical relevance of frequency-difference consensus further in a subsequent paper  (Velasco and Steel, in preparation).

\subsection{Extension of concordance to the supertree setting}

Suppose now that we have as our input a sequence $(T_1, T_2, \ldots, T_k)$ of rooted phylogenetic trees on overlapping leaf sets.  We let $X_i$ denote the leaf set of $T_i$ for each $i$.  
We will say that a particular input tree $T_i$ {\em supports}  the triple $xy|z$ if $T_i$ contains $x$, $y$, and $z$ as leaves, and there is a clade of $T_i$ containing $x$ and $y$ but not $z$. We will say that $T_i$  {\em contradicts} $xy|z$ if $T_i$  supports either $xz|y$ or $yz|x$. Notice that a tree does not support or contradict a particular triple if it does not contain all three taxa as leaves and that it is possible to contain all three taxa and support none of the triples if some splits are unresolved.  

We can then  define a concordance factor using triplet relations as follows:
Let $G(a,b,c) =\{i\in \{1,\ldots, k\}: \{a,b,c\} \subseteq X_i\}$, and for each  $i \in G(a,b,c)$, write $xy|_iz$ if  $T_i|\{a,b,c\} = xy|z$.  Thus, $G(a,b,c)$ is the set of genes (loci) that are present in every one of the three taxa $a, b,c$, and $ab|_ic$ (for instance) means that gene $i$ supports $a$ and $b$ being sister taxa relative to $c$.
Given a subset $A$ of $X = \cup_{i=1}^k X_i$, define the concordance factor of a non-empty subset $A$ of $X$  be:
\begin{itemize}
\item 
$cf(A) := 1$ if $A=X$ or $A= \{x\}$ for some $x \in X$;
\item If $1< |A|<|X|$:
\begin{equation}
\label{supcon}
cf(A):=\frac{  |\{i \in \{1, \ldots, k\}:  aa'|_ib \mbox{ for all } a,a' \in A, b \in X-A \mbox{  with } i \in G(a,a',b) \}|}{|\{i \in \{1, \ldots, k\}:  i \in G(a,a',b)\mbox{ for some } a,a' \in A, b \in X-A\}|},
\end{equation}
provided the denominator is non-zero, otherwise set $cf(A)=0$. 
\end{itemize}

\bigskip

Stated slightly less precisely, for non-singleton proper subset $A$ of $X$
 the concordance factor of $A$ is the proportion of trees $T_i$ having at least two elements from $A$ and  one from outside $A$, for which $aa'|_ib$ for all $a,a' \in A$ and all $b$ not in $A$.  In other words, for a given tree $T_i$ to count positively toward the $cf$ of $A$, any two taxa in  $A$ present in $T_i$ must be more closely related to each other than to any taxon outside $A$ that is present in $T_i$.   This definition of concordance factor generalizes the earlier one, as we now point out. 

\begin{lemo}
If each input tree $T_i$ has the same leaf set, then for any non-empty subset $A$ of $X$, the concordance factor from Eqn. (\ref{supcon}) coincides with the earlier definition of concordance factor as defined in the consensus setting.
\end{lemo}
\begin{proof}
Suppose that $X_i=X$ for all $i$.
If $A=X$ or $A=\{x\}$ for some $x \in X$, then $cf(A) = 1$, as before, so assume that $1<|A|<|X|$.
Then since  $i \in G(a,a',b)$ for all choices of $a,a', b$ the denominator term in $cf(A)$ is $k$, and so:
$$cf(A)=\frac{1}{k}\times  |\{i: \mbox{ for all } a,a' \in A, b \in X-A,  aa'|_ib\}|.$$
Now $A$ is a cluster of $T_i$ precisely if for all $ a,a' \in A, b \in X-A$, we have  $aa'|_ib$, and so 
$cf(A)=\frac{1}{k}\times  |\{i: \mbox{  $A$ is a cluster of $T_i$} \}|$, which coincides with the earlier definition. 
\end{proof}

We can now check the obvious altered definitions of strict, majority, and frequency-difference consensus in this supertree setting. But something significant has happened. It is no longer a guarantee that strict and majority clusters will form a tree. That is, the set of all clusters which have $cf$-value $>0.5$ do not necessarily form a tree; indeed the same can be true even with a $cf$-value equal to 1. A simple example of this is five organisms $a, b, c, d,e$ with two input gene trees: $(ab)e$ and $(ac)d$. Now $cf(ab) = cf(ac) = 1$ but these clusters are incompatible.

However, it is still true that the frequency-difference clusters form a tree, and the resulting
supertree method will satisfy axioms (A1), (A2), and (A3), but not (A4). This is as it should be. Recall that a group $A$ is a frequency-difference cluster precisely if $cf(A) > cf(B)$ for any $B$ that is incompatible with $A$. When $cf$-values are defined in the `supertree' way this definition still leads to a tree as we now show. 
\begin{proposition}
\label{superplu}
The frequency-difference clusters form a hierarchy, and so form a tree.  
\end{proposition}
\begin{proof}
We first establish a simple and general result.  Suppose that $g:2^X \rightarrow \RR$ is any function that assigns a real value to a subset of $X$. Then 
$$H_g: = \{A \subseteq X: g(A)>g(B) \mbox{ for all } B \mbox { that overlap } A\},$$
is a hierarchy, since if we suppose to the contrary that  $H_g$ contains two elements $A$ and $A'$  that overlap, then $g(A) > g(A')$, since  $A \in H_g$ and $A$ overlaps $A'$. Interchanging the roles of $A$ and $A'$, the reverse inequality also
holds, but this clearly is not possible. Thus $H_g$ cannot have two elements that overlap, and so $H_g$ is a hierarchy.

We now apply this general result for function $g(A)=cf(A)$ to deduce that: 
$$H_g: = \{A \subseteq X: g(A)>g(B) \mbox{ for all } B \mbox { that overlap } A\}$$
is a hierarchy.  
This completes the proof. 
\end{proof}

The problem of how to deal with data sets with patchy taxon coverage is of significant biological and mathematical interest. While the natural extensions of many consensus methods will often fail to form a tree in this setting, the frequency-difference method, which is a kind of plurality consensus method, will always yield a tree. Further, it satisfies some extremely plausible axioms (A1)--(A3) for what a supertree method should look like. While it fails to satisfy (A4), we have shown that no method could satisfy this independence condition while simultaneously satisfying (A1)--(A3). We believe that these facts together with its inherent plausibility, make the frequency-difference method worthy of more widespread usage and serious study. 

\section{Acknowledgments}

We thank F.R. McMorris, R.C. Powers, and David Bryant for several helpful comments, particularly concerning Adams consensus. We also thank an (anonymous) reviewer and the editors for  additional comments and advice.

\bibliographystyle{sysbio}
\bibliography{arrow_steel}


\section{Appendix: Proof of Proposition \ref{main}}

We employ a proof by contradiction; that is,  by supposing there were a method satisfying (A1$^-$)--(A4), we derive a contradiction. 

Our argument relies on the existence of a classic combinatorial object called a Steiner triple system (STS). This is a collection of 3-element subsets (called `blocks') from $\{1, 2,\ldots, n\}$ for which 
any two subsets intersect in exactly one point.  When an STS exists, it has exactly $b = \frac{n(n-1)}{6}$ blocks.  It is a  basic result in design theory (a branch of combinatorics \citep{van})
that an STS exists precisely when  the division of $n$ by 6  leaves a remainder of 1 or 3.  In particular, there exists an STS with $n=13$ ($=6 \times 2 +1$) and so with $b = 26$ blocks. 

Let us now suppose we have a method $\psi$ satisfying (A1$^-$)--(A4).   We take the taxon set as $X= \{1,2, \ldots, 13\}$ and we label the 26 blocks of the 
STS as $b_1, b_2, \ldots, b_{26}$.  For each block $b_i$,  let $T_{ij}$ (where $j=1,2,3$) denote the three
possible rooted binary trees we can construct that have the leaf set $b_i$.

Now, let $f: X \rightarrow \{1,2,3\}$ be a selection of one value of $j$ for each $i$, and consider the profile $\cp_f$ of trees $(T_{1f(1)}, T_{2f(2)}, \ldots, T_{26f(26)})$.  Each of these $3^{26}$ possible sequences of $26$ trees will comprise a possible input for $\psi$.

By (A1$^-$), $\psi(\cp_f)$ is a rooted phylogenetic tree, which  we will denote as $T_f$, on the  leaf set $X$, or some subset of these leaves.

By (A4),  taking  the set $Y=b_k$ as our subset of taxa we  obtain:
\begin{equation}
\label{psik}
T_f|b_k \mbox{ equals or refines }\psi(\cp_f|b_k).
\end{equation}
Now,
\begin{equation}
\label{psik2}
\psi(\cp_f|b_k) = \psi((T_{kf(k)})) = T_{kf(k)},
\end{equation}
since the first equality holds by repeated applications of (A3) (it is here that we use the STS property that $|b_j \cap b_k|=1$ for all $j \neq k$), and the second equality holds by (A2) in the special case $k=1$ (i.e. $\psi(\cp') = T$ for $\cp'=(T)$).

Combining (\ref{psik}) and (\ref{psik2}) (and noting that a rooted binary tree on three leaves cannot be further refined), we obtain:
\begin{equation}
T_f|b_k =  T_{kf(k)}.
\label{niceeq}
\end{equation}

Let $T'_f = T_f$ if the latter tree is binary; otherwise, let $T'_f$ denote any binary tree obtained from $T_f$ by resolving it arbitrarily.
Then:
\begin{equation} T'_f|b_k =  T_{kf(k)}.
\label{niceeq2}
\end{equation}

Notice that this implies that the leaf set of $T'_f$ must be all of $X$.
Moreover, Eqn. (\ref{niceeq2}) holds for all $3^{26}$ possible choices for $f$.   This gives us $3^{26}$ rooted binary trees, each on the leaf set $X$ of size 13 (one tree for each choice of $f$).

At this point, we invoke a crucial arithmetic fact: $3^{26}$ is larger than the total number of rooted binary trees on 13 leaves, which is $(23)!! = 1 \times 3 \times \cdots \times 23$.
Thus, by the `pigeonhole principle' \citep{van}, at least two of the binary trees $T'_f$ and $T'_{f'}$ must be equal for some pair $f \neq f'$.  But, by (\ref{niceeq2}),  this implies that $T_{kf(k)}= T_{kf'(k)}$ for
all $k$, and so $f=f'$. This contradiction establishes that the initial assumption of the existence of a method satisfying (A1$^-$)--(A4) is not possible.

\hfill$\Box$

\end{document}